\documentclass[reqno]{amsart}

\usepackage{booktabs} 
\usepackage{IEEEtrantools}
\usepackage{amssymb,latexsym,amsfonts,amsmath}
\usepackage{graphicx}
\usepackage{mathrsfs}
\usepackage{dsfont}

\usepackage{tikz}
\usetikzlibrary{calc,shapes,arrows}
\DeclareMathOperator{\diff}{d}
\usepackage{algorithm}
\usepackage{algorithmic}

\usepackage{xspace}

\newtheorem{theorem}{Theorem}[section]

\newtheorem{problem}[theorem]{Problem}

\newtheorem{corollary}[theorem]{Corollary}

\newtheorem{definition}[theorem]{Definition}

\newtheorem{remark}[theorem]{Remark}
\newtheorem{assumption}{Assumption}
\numberwithin{equation}{section}

\usepackage{multicol}
\usepackage{enumerate}
\usepackage{latexsym}
\usepackage{mathrsfs}
\usepackage{dsfont}
\usepackage{textcomp}
\usepackage{xcolor}

\begin{document}
	
\begin{abstract}
In this paper, we study formal synthesis of control policies for partially observed jump-diffusion systems against complex logic specifications. Given a state estimator, we utilize a discretization-free approach for formal synthesis of control policies by using a notation of control barrier functions without requiring any knowledge of the estimation accuracy. Our goal is to synthesize an offline control policy providing (potentially maximizing) a lower bound on the probability that the trajectories of the partially observed jump-diffusion system satisfy some complex specifications expressed by deterministic finite automata. Finally, we illustrate the effectiveness of the proposed results by synthesizing a policy for a jet engine example. 
\end{abstract}

\title[Synthesis of Partially Observed Jump-Diffusion Systems via Control Barrier Functions
]{Synthesis of Partially Observed Jump-Diffusion Systems via Control Barrier Functions
}

\author{Niloofar Jahanshahi$^{1*\dagger}$}
\author{Pushpak Jagtap$^{2*\dagger}$}
\author{Majid Zamani$^{3,1}$}
\address{$^1$Department of Computer Science, LMU Munich, Germany.}
\email{niloofar.jahanshahi@lmu.de}
\address{$^2$Department of Electrical and Computer Engineering, Technical University of Munich, Germany.}
\email{alavaei@ethz.ch}
\address{$^3$Department of Computer Science, University of Colorado Boulder, USA.}
\email{majid.zamani@colorado.edu}
\maketitle

\section{Introduction}
	Recent years have witnessed a growing interest in formal synthesis of controllers for complex systems against complex logic specifications \cite{belta2017formal}. These specifications are usually expressed using temporal logic formulae or as (in)finite strings over finite automata. Several approaches based on finite abstraction have been widely used to solve such synthesis problems. Existing techniques include policy synthesis enforcing linear temporal logic specifications for non-stochastic systems \cite{tabuada2009verification,belta2017discrete} and for stochastic ones \cite{zamani2014symbolic,zamani2017towards,lavaei2020compositional1}. 
	When dealing with large systems, these approaches suffer severely from the curse of dimensionality (\textit{i.e.}, computational complexity grows exponentially with the dimension of the state set).  In order to overcome the large computational burden, a discretization-free approach, based on control barrier functions has shown potential to solve the formal synthesis problems (See \cite{ames2016control,jagtap2019formal,jagtap2020compositional,huang2017probabilistic} and references therein).
	The aforementioned works assume the availability of complete state information. However, in many real applications we do not have access to complete state information. 
	Motivated by this limitation, the recent result in \cite{clark2019control} provides the synthesis of controllers enforcing invariance properties for stochastic control systems with incomplete information by assuming a prior knowledge of the control barrier functions. 
	In our recent result \cite{CBF1NPM}, we consider the problem of synthesizing controllers for partially observed stochastic control systems. In particular, we search for a control barrier function that provides a controller along with a lower bound on the probability that the system satisfies invariance specifications over a finite-time horizon. Similar to \cite{clark2019control}, this work also assumes the existence of an estimator with a given probabilistic accuracy. Then we provide the overall probability threshold using the probability bound on the estimator accuracy and that of the trajectories of the estimator satisfying the invariance specifications, obtained via control barrier functions. 
	
	The contributions of this paper in comparison with those of \cite{clark2019control,CBF1NPM} are twofold. First, we provide an offline controller synthesis approach enforcing complex logic specifications expressed by (non)deterministic finite automata for partially observed jump-diffusion systems. As a special case, those properties include invariance ones. Second, we provide an approach for computing lower bound on the probability that the system satisfies given specifications over a finite-time horizon \textit{without} requiring any knowledge of the estimator's accuracy. Finally, we demonstrate the effectiveness of the proposed results on a nonlinear jet engine example.    
	\section{Preliminaries and Problem Definition}
	\textbf{Notations:} We denote the set of natural, real, and non-negative real numbers by $\mathbb{N}$, $\mathbb{R}$, and $\mathbb{R}_0^+ $, respectively. We use $\mathbb{R}^n$ to denote the $n$-dimensional Euclidean space and $\mathbb{R}^{n \times r}$ to denote the space of real matrices with $n$ rows and $r$ columns. We denote by $ e_i\in \mathbb{R}^n $ the vector whose all elements are zero, except the $ i^{th} $ element, which is one. Given a matrix $A\in\mathbb{R}^{n\times n}$, $\text{Tr}(A)$ represents trace of $A$ which is the sum of all diagonal elements of $A$.
	The zero matrix in $\mathbb{R}^{n\times m}$ is denoted by $0_{n\times m}$.
	Given sets $X$ and $Y$, we donate $f:X\to Y$ an ordinary map from $X$ to $Y$ and the notation $|X|$ denotes the cardinality of set $X$. 
	\subsection{Partially Observed Jump-Diffusion Systems}\label{secB}
	Let the triplet $(\Omega, \mathcal{F}, \mathbb{P})$ denote a probability space with a sample space $ \Omega $, filtration $ \mathcal{F} $, and the probability measure $ \mathbb{P} $. The filtration $\mathbb{F}= (\mathcal{F}_s)_{s\geq 0}$ satisfies the usual conditions of right continuity and completeness \cite{oksendal2007applied}. Let $ (W_{ks})_{s\geq0} $ be $\bar{r}_k$-dimensional $ \mathbb{F} $-Brownian motions, $k=1,2$. Let $(P_{ks})_{s\geq 0}$ be a $\bar{q}_k$-dimensional $\mathbb{F}$-Poisson processes, with $k=1,2$. We assume that the Poisson processes and Brownian motions are independent of each other. The Poisson process $P_{ks}:= [P_{ks}^1; \cdots; P_{ks}^{\bar{q}_k}]$ models $\bar{q}_k$ kinds of events, $k=1,2$, whose occurrences are assumed to be independent of each other.
	We consider the partially observed jump-diffusion system $($po-JDS$)$, denoted by $\mathcal{S}$, which is described by the following stochastic differential equations (SDE)
	\begin{equation}\begin{aligned}\label{eq:S}
		\mathcal{S}:\begin{cases}
			\diff\xi\!\!\!&=f(\xi,\upsilon)\diff t+g_1(\xi)\diff W_{1t}+r_1(\xi)\diff P_{1t},\\
			\diff y\!\!\!&=h(\xi)\diff t+g_2(\xi)\diff W_{2t}+r_2(\xi)\diff P_{2t},	\end{cases}
\end{aligned}\end{equation}
	where $\xi(t)\in X\subseteq\mathbb{R}^n$ is the value of solution process $\xi$ of $\mathcal{S}$, $\upsilon(t)\in U\subseteq\mathbb{R}^m$ is the input vector, and $y(t)\in\mathbb{R}^p$ is the output vector representing the noisy partial observation at time $t\in\mathbb{R}^+_0$ $\mathbb{P}$-almost surely ($\mathbb{P}$-a.s.). Functions $f:X\times U\to\mathbb{R}^n$, $g_1:X\to\mathbb{R}^{n\times \bar{r}_1}$, $g_2:X\to\mathbb{R}^{p\times \bar{r}_2}$, $r_1:X\to\mathbb{R}^{n\times \bar{q}_1}$, $r_2:X\to\mathbb{R}^{p\times \bar{q}_2}$, and  $h:X\to\mathbb{R}^p$ are assumed to be Lipschitz continuous to ensure existence and uniqueness of the solution of $\mathcal{S}$ \cite{oksendal2007applied}. Throughout the paper, we use the notation $\xi_{a\upsilon}(t)$ to denote the value of the solution process of $\mathcal{S}$ at time $ t\in \mathbb{R}_0^+$ under the input signal $\upsilon$ starting from the initial state $\xi_{a\upsilon}(0)=a$  $\mathbb{P}$-a.s., in which $a$ is a random variable that is measurable in $\mathcal{F}_0$. Here, we assume that the Poisson processes $P_{ks}^i$ for any $i\in\{1,\ldots,\bar{q}_k\}$, $k=1,2$, have the rates of $\lambda_{ki}$.
	In order to provide the results in this paper, we raise the following assumption on the existence of the estimator that estimates the state of the po-JDS \eqref{eq:S}.
	\begin{assumption}\label{estimate}
		The states of the po-JDS $\mathcal{S}$ in \eqref{eq:S} can be estimated by a proper estimator $\hat{\mathcal{S}}$ represented in the form of an SDE as:
		\begin{equation}\label{eq:estimator1}
		\hat{\mathcal{S}}:	
		\diff\hat{\xi}=f(\hat{\xi},\upsilon)\diff t+K\big{(}\diff y-h(\hat{\xi})\diff t\big{)},	
		\end{equation}
		where $K\in\mathbb{R}^{n\times p}$ is the estimator gain.
	\end{assumption}
	There are plenty of results in the literature on the computation of estimator gain $K$ for various classes of stochastic systems; see the results in \cite{kai2009robust,clark2019control,chen2008robust}, and \cite{tseng2007robust}.
	We define the augmented process $[\xi, \hat{\xi}]^T$, where $\xi$ and $\hat{\xi}$ are the solution processes of $\mathcal{S}$ and $\hat{\mathcal{S}}$, respectively. The corresponding augmented jump-diffusion system $\tilde{\mathcal{S}}$ can be defined as:
\begin{equation}\begin{aligned}
		\label{eq:augmented}
		&\hspace{-.2em}\begin{bmatrix}
			\diff \xi\\ \diff \hat{\xi}
		\end{bmatrix}\hspace{-.2em}=\hspace{-.2em}\Big{(}\begin{bmatrix}
			f(\xi,\upsilon)\\f(\hat{\xi},\upsilon)
		\end{bmatrix}+\begin{bmatrix}
			0_{n\times p} & 0_{n\times p} \\ K & -K
		\end{bmatrix}\begin{bmatrix}
			h(\xi) \\h(\hat{\xi})
		\end{bmatrix}\Big{)}\diff t
		\nonumber \\ &\hspace{-.1em}+\hspace{-.2em}\begin{bmatrix}
			g_1(\xi) & \hspace{-.6em}0_{n\times\bar{r}_2}\\
			0_{n\times \bar{r}_1 } & \hspace{-.6em}Kg_2(\xi)
		\end{bmatrix}\hspace{-.2em}\begin{bmatrix}
			\diff W_{1t}\\
			\diff W_{2t}
		\end{bmatrix}
		\hspace{-.2em}+\hspace{-.2em}\begin{bmatrix}
			r_1(\xi)\\
			0_{n\times \bar{q}_1}  
			\hspace{-.1em}\end{bmatrix}\hspace{-.1em}\diff P_{1t}
		\hspace{-.2em}+\hspace{-.2em}\begin{bmatrix}
			0_{n\times \bar{q}_2}\\
			Kr_2(\xi) 
		\end{bmatrix}\hspace{-.1em} \diff P_{2t}.
\end{aligned}\end{equation}
	For later use, we provide the definition of the infinitesimal generator (denoted by operator $\mathcal{D}$) for $\tilde{\mathcal{S}}$ using Ito's differentiation \cite{oksendal2007applied}. 
	Let $B:X\times X\rightarrow\mathbb{R}$ be a twice differentiable function. The infinitesimal generator of $B$ associated with the system $\tilde{\mathcal{S}}$ for all $(x,\hat{x})\in X\times X$ and for all $u\in U$ is given by 
	\begin{equation}\begin{aligned}
		\mathcal{D}&B(x,\hspace{-.1em}\hat{x},\hspace{-.1em}u)\hspace{-.2em}=\hspace{-.2em}\begin{bmatrix}
			\partial_x B &\hspace{-.6em}\partial_{\hat{x}}B
		\end{bmatrix} (\begin{bmatrix}
			f(x,u)\\
			f(\hat{x},u)
		\end{bmatrix}
		\hspace{-.2em}+\hspace{-.2em}\begin{bmatrix}
			0_{n\times p} & \hspace{-.6em}0_{n\times p} \\ K & \hspace{-.6em}-K
		\end{bmatrix}\hspace{-.2em}\begin{bmatrix}
			h(x) \\h(\hat{x})
		\end{bmatrix})
		\nonumber\\&+\hspace{-.2em}\frac{1}{2}\text{Tr}(\begin{bmatrix}
			g_1(x) & \hspace{-.6em}0_{n\times\bar{r}_2}\\
			0_{n\times \bar{r}_1 } &\hspace{-.6em} Kg_2(x)
		\end{bmatrix}
		\hspace{-.2em}\begin{bmatrix}
			g_1(x) & \hspace{-.6em}0_{n\times\bar{r}_2}\\
			0_{n\times \bar{r}_1 } &\hspace{-.6em} Kg_2(x)
		\end{bmatrix}^T\hspace{-.2em}
		\begin{bmatrix}
			\partial_{xx}B & \hspace{-.6em}\partial_{x\hat{x}}B\\
			\partial_{\hat{x}x}B & \hspace{-.6em}\partial_{\hat{x}\hat{x}}B
		\end{bmatrix})
		\nonumber\\&+\hspace{-.2em}\sum_{i=1}^{\bar{q}_1}\lambda_{1i}(B(x+r_1(x)\text{e}_i,\hat{x})-B(x,\hat{x}))
		\nonumber\\&+\sum_{i=1}^{\bar{q}_2}\lambda_{2i}(B(x+Kr_2(x)\text{e}_i,\hat{x})-B(x,\hat{x})).
		\label{eq:D}
	\end{aligned}\end{equation}
	The symbols $\partial_{x}$ and $\partial_{x,\hat{x}}$ in \eqref{eq:D} represent first and second-order partial derivatives with respect to $x$ (1st argument) and $\hat{x}$ (2nd argument), respectively. Note that we dropped the arguments of $\partial_{x}B$, $\partial_{\hat x}B$, $\partial_{x,x}B$, $\partial_{x,\hat x}B$, $\partial_{\hat x, x}B$, and $\partial_{\hat x,\hat x}B$ in \eqref{eq:D} for the sake of simplicity.
	
	Given a po-JDS $\mathcal{S}$ in \eqref{eq:S}, we aim at synthesizing a control policy that guarantees a potentially tight lower bound on the probability that system $\mathcal{S}$ satisfies a complex specification over a finite time horizon. The class of specifications considered in this paper are provided in the next subsection. 
	\begin{remark}
		The use of the augmented system $\tilde{\mathcal{S}}$ will allow us to provide the main result of the paper without any correctness requirement on the observer. In particualr,
			our augmented system formulation provides the user the flexibility to design any observer by means of any technique. The probabilistic distance
			between the values of state and their estimator is natively considered in
			our formulation and one does not need to quantify this distance a-priori
			which is needed in the results proposed in \cite{CBF1NPM,clark2019control}.
	\end{remark}
	\subsection{Specifications}
	In this subsection, we consider the class of specifications expressed by nondeterministic finite automata $($NFA$)$ as defined below.
	\begin{definition}\cite{baier2008principles}\label{NFA} 
		A nondeterministic finite automaton $($NFA$)$ is a tuple $\mathcal{A}=(Q,Q_0,\Sigma,\delta,F)$, where $Q$ is a finite set of states, $Q_0\subseteq Q$ is a set of initial states, $\Sigma$ is a finite set $($a.k.a. alphabet$)$, $\delta: Q\times\Sigma\rightarrow P(Q)$ is a transition function, where $P(Q)$ denotes the power set of $Q$, and $F\subseteq Q$ is a set of accepting (or final) states.
	\end{definition}
	
	NFA $\mathcal{A}$ is called \textit{deterministic} if the transition function is defined as $\delta:Q\times \Sigma\to Q$, and we refer to it as deterministic finite automata (DFA). Since every NFA can be converted to its equivalent DFA using the powerset construction \cite{bonchi2013checking}, in the rest of the paper, we only deal with DFA. Moreover, it is well known that the complement of a DFA $\mathcal{A}$, denoted by $\mathcal{A}^c$, is again a DFA \cite{hopcroft2001introduction}.
	We use the notation $q\overset{\sigma}{\longrightarrow} q'$ to denote transition relation $(q,\sigma,q')\in\delta$. A finite word $\sigma=(\sigma_0,\sigma_1,\ldots,\sigma_{k-1})\in \Sigma^k$ is accepted by DFA $\mathcal{A}$ if there exists a finite state run $\mathbf{q}=(q_0,q_1,\ldots,q_{k})\in Q^{k+1}$ such that $q_0\in Q_0$, $q_i \overset{\sigma_i}{\longrightarrow} q_{i+1}$ for all $0\leq i< k$ and $q_{k}\in F$. The accepted language of $\mathcal{A}$, denoted by $\mathcal{L}(\mathcal{A})$, is the set of all words accepted by $\mathcal{A}$.
	
	In this work, we consider those specifications given by the accepting languages of DFA $\mathcal{A}$ defined over a set of atomic propositions $\Pi$, i.e., the alphabet $\Sigma=\Pi$.
	We should highlight that all linear temporal logic specifications defined over finite traces, referred to as LTL$_F$, are recognized by DFA \cite{de2015synthesis}.   
	
	\subsection{Satisfaction of Specification by po-JDS}
	A given po-JDS $\mathcal{S}$ in \eqref{eq:S} is connected to the specification given by the accepting language of a DFA $\mathcal{A}$ defined over a set of atomic propositions $\Pi$, with the help of a measurable labeling function $L: X \rightarrow \Pi$ as described in the next definition which is similar to  \cite[Definition 2]{wongpiromsarn2015automata}.
	\begin{definition}\label{Def_Finite_trace}
		For a po-JDS $\mathcal{S}$ as in \eqref{eq:S} and the labeling function $L:X \rightarrow \Pi$, a finite sequence $\sigma(\xi_{a \upsilon})=(\sigma_0,\sigma_1,\ldots,\sigma_{k-1}) \in \Pi^{k}$, $k\in\mathbb{N}$, is a finite trace of the solution process $\xi_{a\upsilon}$ over a finite time horizon $[0,T)\subset\mathbb{R}_0^+$ if there exists an associated time sequence $t_0, t_1,\ldots,t_{k-1}$ such that $t_0=0$, $t_k=T$, and for all $j \in \{0,1,\ldots,k-1\}$, $t_j \in \mathbb{R}_{0}^+$ following conditions hold
		\begin{itemize}
			\item $t_j<t_{j+1}$;
			\item $\xi_{a\upsilon}(t_j)\in L^{-1}({\sigma_j})$;
			\item If $\sigma_j \neq \sigma_{j+1}$, then for some $t_j' \in [t_j,t_{j+1}]$, $\xi_{a\upsilon}(t)\in L^{-1}({\sigma_j})$ for all $t \in (t_j,t_j')$; $\xi_{a\upsilon}(t)\in L^{-1}({\sigma_{j+1}})$ for all $t \in (t_j',t_{j+1})$; and either $\xi_{a\upsilon}(t_j')\in L^{-1}(\sigma_j)$ or $\xi_{a\upsilon}(t_j')\in L^{-1}(\sigma_{j+1})$.
		\end{itemize}
	\end{definition}
	Next, we define the probability that the solution process $\xi_{a\upsilon}$ of the po-JDS $\mathcal{S}$ starting from some initial state $\xi_{a\upsilon}(0)=a\in X_0$ under control policy $\upsilon$ satisfies
	the specification given by DFA $\mathcal{A}$.
	\begin{definition}
			The finite trace corresponding to the solution process of a po-JDS $\mathcal{S}$ starting from $a\in X$ and under the control policy $\upsilon$ over a finite-time horizon $[0,T)\subset \mathbb{R}^+_0$, i.e. $\sigma(\xi_{a\upsilon})=(\sigma_0,\sigma_1,\ldots,\sigma_j,\ldots,\sigma_{k-1})\in \Pi^k$ as in Definition \ref{Def_Finite_trace}, satisfies a specification given by the language of a DFA $\mathcal{A}$, denoted by $\sigma(\xi_{a\upsilon})\models \mathcal{A} $, if there exists $j\in \{0,\ldots,k-1\}$ such that $(\sigma_0,\sigma_1,\ldots,\sigma_j)\in \mathcal{L}(\mathcal{A})$. The probability of satisfaction of the specification given by $\mathcal A$ is denoted by $\mathbb{P}\{\sigma(\xi_{a\upsilon})\models\mathcal{A}\}$.
	\end{definition}
	
	\begin{remark}
		The set of atomic propositions $\Pi=\{p_0,p_1,\ldots,p_M\}$ and the labeling function $L: X \rightarrow \Pi$ provide a measurable partition of the state set $X = \cup_{i=1}^N X_i$ as  $X_i:=L^{-1}(p_i)$. Without loss of generality, we assume that $X_i\neq \emptyset$ for any $i$.
	\end{remark}
	\subsection{Problem Definition}
	Now, we formally define the main synthesis problem considered in this work. 
	\begin{problem}\label{problem}
		Given a po-JDS $\mathcal{S}$ as in \eqref{eq:S}, a specification given by the accepting language of DFA $\mathcal{A}=(Q,Q_0,\Pi,\delta,F)$ over a set of atomic propositions $\Pi=\{p_0,p_1,\ldots,p_M\}$, a labeling function $L: X \rightarrow \Pi$, and a real value $\vartheta\in(0,1)$, compute an offline control policy $\upsilon$ $($if existing$)$ such that $\mathbb{P}\{\sigma(\xi_{a\upsilon})\models\mathcal{A}\}\geq \vartheta$, for all $a\in L^{-1}(p_i)$ and some $i\in\{0,1,\ldots,M\}$.
	\end{problem}
	
	Finding a solution to Problem \ref{problem} (if existing) is difficult in general. We should highlight that the proposed approach here is sound in solving the considered synthesis problem. This means that if the proposed method provides a solution to a synthesis problem, then we can formally conclude that the proposed controller renders the given specification with the corresponding lower bound on the probability of satisfaction. However, if the method fails to provide any solution, then there may or may not exist a solution to the original synthesis problem). Our approach is to compute a policy $\upsilon$ together with a lower bound \underline{$\vartheta$}. Our aim is to find the potentially largest lower bound, which can be compared with $\vartheta$ and gives policy, i.e., a solution for Problem \ref{problem} if \underline{$\vartheta$}$\geq\vartheta$. Instead of computing a control policy that guarantees the lower bound
	\underline{$\vartheta$}, we compute a policy that guarantees $\mathbb{P}\{\sigma(\xi_{a\upsilon})\models\mathcal{A}^c\}
	\leq\bar{\vartheta}$, for any $a\in L^{-1}(p_i)$ and some $i\in\{0,1,\ldots,M\}$. Then for the same control policy the lower bound can be easily obtained as \underline{$\vartheta$}$=1-\bar{\vartheta}$. This is done by constructing a DFA $\mathcal{A}^c$ whose language is the complement of the language of DFA $\mathcal{A}$. To synthesize a controller, we utilize the notion of control barrier functions defined for augmented jump-diffusion system $\tilde{\mathcal{S}}$ introduced in the next section. 
	\section{Control Barrier Functions}
	In this section, we provide sufficient conditions using so-called control barrier functions under which we can provide the  upper  bound  on  the  probability  that  the  trajectories of system $\mathcal{S}$ starting from any initial state in $X_0\subseteq X$  reach $X_1\subseteq X$.
	To provide a result giving an upper bound on the reachability probability for the trajectory of $\mathcal{S}$, we provide conditions on barrier functions constructed over the augmented system $\tilde{\mathcal{S}}$.
	\begin{theorem} \label{barrier1}
		Consider a po-JDS $\mathcal{S}$ as in \eqref{eq:S}, its estimator $\hat{\mathcal{S}}$ as in \eqref{eq:estimator1}, the resulting augmented system $\tilde{\mathcal{S}}$ as in \eqref{eq:augmented} and sets $X_0, X_1\subseteq X $. Suppose there exists a twice differentiable function $B: X\times X\rightarrow \mathbb R_0^{+}$, constants $c\geq 0$ and $\gamma \in [0,1)$ such that
		\begin{equation} 
			\forall (x,\hat{x}) \in X_0\times X_0,\quad \  \quad \quad B(x,\hat{x}) \leq \gamma, \label{eq:th1}
			\end{equation}
		\begin{equation}
			\forall (x,\hat{x}) \in X_1 \times X, \ \ \quad \quad \quad B(x,\hat{x}) \geq 1,  \label{eq:th2}
		\end{equation}
	\begin{equation}
			\forall \hat{x}\in X, \exists u\in U, \forall x\in X, \quad \mathcal{D}B(x,\hat{x},u)\leq c.
			\label{eq:th3}
		\end{equation}
		Then the probability that the solution process $\xi_{a\upsilon}$ of the system $\mathcal{S}$ starts from any initial state $a\in X_0$ and reaches region $X_1$ under the control policy $\upsilon$ within time horizon $[0,T)\subset\mathbb{R}_0^+$ is upper bounded by $\gamma+c T$.
	\end{theorem}
	\begin{proof}
		By using \eqref{eq:th1} and the fact that $ X_1\times X\subseteq\big{\{}(x,\hat{x})\in X\times X \mid B(x,\hat{x})\geq 1\big{\}}$, we have
		$\mathbb{P}\big{\{}\xi_{a\upsilon}(t)\in X_1\wedge\hat{\xi}_{\hat{a}\upsilon}(t)\in X ~\exists t\in[0,T) \mid a, \hat{a}\big{\}}\leq
		\mathbb{P}\big{\{}\text{sup}_{0\leq t\leq T}B(\xi_{a\upsilon}(t),\hat{\xi}_{\hat{a}\upsilon}(t)) \geq 1 \mid a,\hat{a}\big{\}}\leq
		B(a,\hat{a})+cT\leq \gamma +cT$. The second inequality is obtained by utilizing the result of \cite[Theorem 1]{kushner1967stochastic}. This implies that the probability of the augmented trajectory of $\tilde{\mathcal{S}}$ staring from any $(a,\hat a)\in X_0\times X_0$ and reaching $X_1\times X$ is upper bounded by $\gamma +cT$.\\
		Now we get $\mathbb{P}\big{\{}\xi_{a\upsilon}(t)\in X_1\wedge\hat{\xi}_{\hat{a}\upsilon}(t)\in X ~\exists t\in[0,T) \mid a, \hat{a}\big{\}}\leq\mathbb{P}\big{\{}\xi_{a\upsilon}(t)\in X_1 ~\exists t\in[0,T) \mid a\big{\}}+
		\mathbb{P}\big{\{}\hat{\xi}_{\hat{a}\upsilon}(t)\in X ~\exists t\in[0,T) \mid \hat a\big{\}}-
		\mathbb{P}\big{\{}\xi_{a\upsilon}(t)\in X_1 \vee \hat{\xi}_{\hat{a}\upsilon}(t)\in  X ~\exists t\in[0,T) \mid a, \hat{a}\big{\}}$.
		Since, the second and last terms trivially hold with probability 1, one has
		$\mathbb{P}\big{\{}\xi_{a\upsilon}(t)\in X_1\wedge\hat{\xi}_{\hat{a}\upsilon}(t)\in X ~\exists t\in[0,T) \mid a, \hat{a}\big{\}}\leq\mathbb{P}\big{\{}\xi_{a\upsilon}(t)\in X_1 ~\exists t\in[0,T) \mid a\big{\}}$.
		Now, since the right term of the \texttt{and} (i.e. $\wedge$) is held for all time, the inequality above becomes an equality and one gets
		$\mathbb{P}\big{\{}\xi_{a\upsilon}(t)\in X_1 ~\exists t\in[0,T) \mid a\big{\}} \leq \gamma+Tc$
		which concludes the proof. 
	\end{proof} 
	The function $B$ in Theorem \ref{barrier1} satisfying \eqref{eq:th1}-\eqref{eq:th3} is usually referred to as the control barrier function.
	\begin{remark}\label{reamrk1}
		Condition \eqref{eq:th3} implicitly associates a stationary controller $\mathsf{u} : X \rightarrow U$ according to the existential quantifier on $u$ for any $\hat x \in X$ and is independent of choice of $x\in X$. The stationary control policy $\upsilon$ driving the system is readily given by $\upsilon(t) = \mathsf{u}(\hat\xi_{a\upsilon}(t))$,
		where $\hat\xi_{a\upsilon}$ is the solution process of the estimator.
	\end{remark}
	\section{Formal Synthesis of Controllers}
	To synthesize control policies using control barrier functions enforcing specifications expressed by DFA $\mathcal{A}$, we first provide the decomposition of specifications into sequential reachability tasks which will later be solved using control barrier functions.
	\subsection{Decomposition into Sequential Reachability}
	\label{subsection:runs}
	Consider a DFA $\mathcal{A}$ expressing the properties of interest for the system $\mathcal{S}$. Consider DFA $\mathcal{A}^c=(Q,Q_0,\Pi,\delta,F)$ whose language is the complement of the language of DFA $\mathcal{A}$.
	The sequence $\mathbf{q}=(q_0,q_1,\ldots,q_k)\in Q^{k+1}$, $k\in\mathbb{N}$ is called an accepting state run if $q_0\in Q_0$, $q_k\in F$, and there exists a finite word $\sigma = (\sigma_0,\sigma_1,\ldots,\sigma_{k-1})\in\Pi^k$ such that $q_i \overset{\sigma_i}{\longrightarrow} q_{i+1}$ for all $i\in\{0,1,\ldots, k-1\}$. We denote the finite word corresponding to accepting state run $\mathbf{q}$ by $\sigma(\mathbf{q})$.
	We also indicate the length of $\mathbf{q}\in Q^{k+1}$ by $|\mathbf{q}|$, which is $k+1$. Let $\mathcal{R}$ be the set of all finite accepting state runs starting from $q_0\in Q_0$ excluding self-loops, where
	$$\mathcal{R}\hspace{-.2em} := \hspace{-.2em} \{\mathbf{q}\hspace{-.2em}  =\hspace{-.2em}  (q_0,q_1,\ldots,q_k)\hspace{-.2em} \in \hspace{-.2em} Q^{k+1} \mid q_k\hspace{-.2em} \in \hspace{-.2em} F, q_i\hspace{-.2em} \neq\hspace{-.2em}  q_{i+1},\forall i\hspace{-.2em} <\hspace{-.2em} k\}.$$
	Computation of $\mathcal{R}$ can be done algorithmically by viewing $\mathcal{A}^c$ as a directed graph $\mathcal{G}=(\mathcal{V},\mathcal{E})$ with vertices $\mathcal{V}=Q$ and edges $\mathcal{E}\subseteq\mathcal{V}\times\mathcal{V}$ such that $(q,q')\in\mathcal{E}$ if and only if $q'\neq q$ and there exist $p\in\Pi$ such that $q\overset{p}{\longrightarrow} q'$. For any $(q,q')\in \mathcal{E}$, we donate the atomic proposition associated with the edge $(q,q')$ by $\sigma(q,q')$.
	From the construction of the graph, it is obvious that the finite path in the graph starting from vertices $q_0\in Q_0$ and ending at $q_F\in F$ is an accepting state run $\mathbf{q}$ of $\mathcal{A}^c$ without any self-loop and therefore belongs to $\mathcal{R}$. One can easily compute $\mathcal{R}$ using depth first search algorithm \cite{russell2003artificial}.
	For each $p \in \Pi$, we define a set $\mathcal{R}^p$ as
	\begin{equation}
	\label{eq:runs}
	\mathcal{R}^p := \{\mathbf{q} = (q_0,q_1,\ldots,q_k)\in\mathcal R \mid \sigma(q_0,q_1)=p\}.
	\end{equation}
	Decomposition into sequential reachability is performed as follows. For any $\mathbf{q}=(q_0,q_1,\ldots,q_k) \in \mathcal{R}^p \ \forall p \in \Pi$, we define $\mathcal{P}^p(\mathbf{q})$ as a set of all state runs of length $3$,
	\begin{equation} \label{eq:reachability}
	\mathcal{P}^p(\mathbf{q}):=\{(q_i,q_{i+1},q_{i+2}) \mid 0 \leq i \leq k-2\}.
	\end{equation}
	Now, we define $\mathcal{P}(\mathcal{A}^c):=\bigcup_{p\in\Pi}\bigcup_{\mathbf{q}\in\mathcal{R}^p}\mathcal{P}^p(\mathbf{q})$.
	\begin{remark}
		Note that $\mathcal{P}^p(\mathbf{q})=\emptyset$ for $|\mathbf{q}|=2$. In fact, any accepting state run of length $2$ specifies a subset of the state set such that the system satisfies $\mathcal{A}^c$ whenever it starts from that subset. This gives trivial zero probability for satisfying the specification, thus neglected in the sequel.
	\end{remark}
	
	For the illustration of the above sets, we kindly refer the interested reader to Example 1 in \cite{jagtap2019formal}.
	Having $\mathcal{P}^p(\mathbf{q})$ in \eqref{eq:reachability} as the set of state runs of length $3$, in this subsection, we provide a systematic approach to compute a policy together with a (potentially tight) lower bound on the probability that the solution process of $\mathcal{S}$ satisfies the specifications given by DFA $\mathcal{A}$. Given a DFA $\mathcal{A}^c$, our approach relies on performing a reachability computation over each element of $\mathcal{P}(\mathcal{A}^c)$  (i.e., $\bigcup_{p\in\Pi}\bigcup_{\mathbf{q}\in\mathcal{R}^p}\mathcal{P}^p(\mathbf{q})$), where reachability probability is upper bounded using control barrier functions along with appropriate choices of control inputs as mentioned in Theorem~\ref{barrier1}. However, computation of control barrier functions and the policies for each element $\nu\in\mathcal P(\mathcal{A}^c)$, can cause ambiguity while utilizing controllers in closed-loop whenever there are more than one outgoing edges from a state of the automaton. 
	To resolve this ambiguity, we simply merge such reachability problems into one reachability problem by replacing the reachable set $X_1\times X$ in Theorem~\ref{barrier1} with the union of regions corresponding to the alphabets of all outgoing edges. Thus we get a common control barrier function and a corresponding controller. This enables us to partition $\mathcal P(\mathcal{A}^c)$ and put the elements sharing a common control barrier function and a corresponding controller in the same partition set. These sets can be formally defined as
	\begin{equation*}\begin{aligned}
		\mu_{(q,q',\Delta(q'))}:=\{(q,q',&q'')\in\mathcal P(\mathcal{A}^c)\\&\mid q,q',q''\in Q\text{ and }q''\in\Delta(q')\}.
	\end{aligned}\end{equation*}
	The control barrier function and the controller (as discussed in Remark \ref{reamrk1}) corresponding to the partition set $\mu_{(q,q',\Delta(q'))}$ are denoted by $B_{\mu_{(q,q',\Delta(q'))}}(x,\hat{x})$ and $\mathsf{u}_{\mu_{(q,q',\Delta(q'))}}(\hat{x})$, respectively. Thus, for all $\nu\in\mathcal P(\mathcal{A}^c)$, we have 
	\begin{equation}\begin{aligned}\label{eq:controller}
		B_\nu(x,\hat{x})=B_{\mu_{(q,q',\Delta(q'))}}(x,\hat{x})\text{ and } \mathsf{u}_\nu(\hat{x})=\mathsf{u}_{\mu_{(q,q',\Delta(q'))}}(\hat{x}),\nonumber\\
		\text{if } \nu\in\mu_{(q,q',\Delta(q'))}.
	\end{aligned}\end{equation}
	\subsection{Control Policy}\label{aaaaa1}
	From the above discussion, one can readily observe that we have different control policies at different locations of the automaton which can be interpreted as a switching control policy.
	Next, we define the automaton representing the switching mechanism for control policies. Consider the DFA $\mathcal{A}^c=(Q,Q_0,\Pi,\delta,F)$ corresponding to the complement of DFA $\mathcal{A}$
	as discussed in Section~\ref{subsection:runs}, where $\Delta(q)$ denotes the set of all successor states of $q\in Q$. Now, the switching mechanism is given by a DFA $\mathcal{A}_{\mathfrak m}=(Q_{\mathfrak m},Q_{\mathfrak m 0},\Pi_{\mathfrak m},\delta_{\mathfrak m},F_{\mathfrak m})$, where $Q_{\mathfrak m}:=Q_{\mathfrak m 0}\cup\{(q,q',\Delta(q'))\mid q,q'\in Q\setminus F\}\cup F_{\mathfrak m}$ is the set of states, $Q_{\mathfrak m 0}:=\{(q_0,\Delta(q_0))\mid q_0\in Q_0\}$ is the set of initial states, $\Pi_{\mathfrak m}=\Pi$, $F_{\mathfrak m}=F$, and the transition relation $(q_{\mathfrak m},\sigma,q_{\mathfrak m}')\in \delta_{\mathfrak m}$ is defined as
	\begin{itemize}
		\item for all $q_{\mathfrak m}=(q_0,\Delta(q_0))\in Q_{\mathfrak m 0}$,\\ $(q_0,\hspace{-.1em}\Delta(q_0))\hspace{-.2em}\overset{\sigma(q_0,q'')}{\longrightarrow}\hspace{-.2em}(q_0,\hspace{-.1em}q'',\hspace{-.1em}\Delta(q''))$, where $q_0\hspace{-.2em}\overset{\sigma(q_0,q'')}{\longrightarrow}\hspace{-.2em}q''$;
		\item for all $q_{\mathfrak m}=(q,q',\Delta(q'))\in Q_{\mathfrak m}\setminus (Q_{\mathfrak m 0}\cup F_{\mathfrak m})$,
		\begin{itemize}
			\item $(q,q',\Delta(q'))\overset{\sigma(q',q'')} {\longrightarrow}(q',q'',\Delta(q''))$, such that $q,q',q''\in Q$, $q'\!\!\overset{\sigma(q',q'')}{\longrightarrow}q''$, and $q''\notin F$; and
			\item $(q,q',\Delta(q'))\overset{\sigma(q',q'')} {\longrightarrow} q''$, such that $q,q',q''\in Q$, $q'\!\!\overset{\sigma(q',q'')}{\longrightarrow}q''$, and $q''\in F$.
		\end{itemize}
	\end{itemize}
	The hybrid controller defined over augmented state-space $X\times Q_\mathsf{m}$ that is a candidate for solving Problem~\ref{problem} is given by
	\begin{equation}\label{eq:policy}
	\tilde{\mathsf u}(\hat{x},q_{\mathfrak m})=\mathsf{u}_{\mu_{(q_{\mathfrak m}')}}(\hat{x}), \quad \forall (q_{\mathfrak m},L(\hat x),q_{\mathfrak m}')\in\delta_{\mathfrak m}.
	\end{equation}
	The corresponding hybrid control policy $\upsilon$ is given by $\upsilon(t) = \tilde{\mathsf u}(\hat\xi(t), q_\mathsf{m})$.
	For the illustration of the switching mechanism, see Example 1 in \cite[Section 5]{jagtap2019formal}. In the next subsection, we discuss the computation of bound on the probability of satisfying the specification under such a policy, which then can be used for checking if this policy is indeed a solution for Problem~\ref{problem}.
	\subsection{Computation of Probability}
	The next theorem provides an upper bound on the probability that the solution process satisfies the specifications given by $\mathcal{A}$.
	\begin{theorem} \label{upper_bound}
		For a specification given by the accepting language of DFA $\mathcal{A}$, let $\mathcal{A}^c$ be the DFA corresponding to the complement of $\mathcal{A}$, $\mathcal R^p$ be the set defined in \eqref{eq:runs}, and $\mathcal{P}^p$ be the set of runs of length $3$ defined in \eqref{eq:reachability}.  Then the probability that the solution process of the system $\mathcal{S}$ starting from any initial state $a \in L^{-1}(p)$ under the hybrid control policy $\upsilon$ associated with the hybrid controller \eqref{eq:policy} satisfies $\mathcal{A}^c$ within time horizon $[0,T)$ is upper bounded by
		\begin{equation}\begin{aligned}\begin{split}
				\label{eq:finalprob}
				\mathbb{P}\{\sigma(\xi_{a\upsilon})\hspace{-0.3em}\models\hspace{-0.3em}\mathcal{A}^c\}\hspace{-0.3em} \le \hspace{-0.5em}\sum_{\bf q \in \mathcal{R}^p}\hspace{-0.4em}\prod\{\hspace{-0.1em}(\gamma_\nu\hspace{-0.2em}+\hspace{-0.2em}c_\nu T)\hspace{-0.2em} \mid\hspace{-0.2em} \nu\hspace{-0.2em}=\hspace{-0.2em}(q,\hspace{-0.1em}q',\hspace{-0.1em}q'') \hspace{-0.2em}\in\hspace{-0.2em} \mathcal{P}^p(\bf{q})\hspace{-0.1em}\},
			\end{split}
		\end{aligned}\end{equation}
		where $\gamma_\nu+c_\nu T$ is the upper bound on the probability that the solution process of $\mathcal{S}$ starts from $X_0 := L^{-1}(\sigma(q,q'))$ and reaches $ X_1:=L^{-1}(\sigma(q',q''))$ under control policy $\upsilon$ within time horizon $[0,T)$ which is computed via Theorem~\ref{barrier1}. 
	\end{theorem}
	\begin{proof} 
		The proof is similar to that of \cite[Theorem 5.2]{jagtap2019formal} and is omitted here due to the lack of space. 
	\end{proof}
	
	Theorem~\ref{upper_bound} enables us to decompose the specification 
	into a collection of sequential reachabilities, compute bounds on the reachability probabilities using Theorem~\ref{barrier1},
	and then combine the bounds in a sum-product expression.
	\begin{remark}
		In case we are unable to find control barrier functions for some of the elements $\nu \in \mathcal P^p(\mathbf q)$ in \eqref{eq:finalprob}, we replace the related term $(\gamma_\nu+c_\nu T)$ by the pessimistic bound $1$ and apply random control input. In order to get a non-trivial bound in \eqref{eq:finalprob}, at least one control barrier function must be found for each $\mathbf{q} \in\mathcal R^p$.
	\end{remark}
	\begin{corollary}
		\label{lower_bound}
		Given the result of Theorem \ref{upper_bound}, the probability that the solution process of $\mathcal{S}$ starts from any $a \in L^{-1}(p)$ under control policy $\upsilon$ and satisfies specifications given by DFA $\mathcal{A}$ over time horizon $[0,T)\subset\mathbb{R}_0^+$ is lower-bounded by
		\begin{equation}
		\mathbb{P}\{\sigma(\xi_{a\upsilon}) \models \mathcal{A}\}\geq1-\mathbb{P}\{\sigma(\xi_{a\upsilon}) \models \mathcal{A}^c\}.\nonumber
		\end{equation}
	\end{corollary}
	\subsection{Computation of Control Barrier Functions} \label{section:Computation}
	{Proving the existence of a control barrier function and finding one are in general hard problems. 
		However, if functions $f$, $h$, $g_1$, $g_2$, $r_1$, and $r_2$ are polynomial with respect to their arguments and partition sets $X_i = L^{-1}(p_i), i \in \{0,1,2,\ldots,M\}$, are bounded semi-algebraic sets (i.e., they can be represented by polynomial (in)equalities), one can formulate conditions in Theorem~\ref{barrier1} as a sum-of-squares (SOS) optimization problem. See \cite[Section 5.3.1.]{jagtap2019formal} for a detailed discussion on a similar approach. Having an SOS optimization problem, one can efficiently search for a polynomial control barrier function $B_{\nu}(x,\hat{x})$ and controller $\mathsf{u}_{\nu}(\hat{x})$, for any $\nu \in \mathcal P(\mathcal{A}_{\neg\varphi})$ as in \eqref{eq:controller} using SOSTOOLS \cite{prajna2002introducing} in conjunction with a semidefinite programming solver such as SeDuMi \cite{sturm1999using} while minimizing constants $\gamma_{\nu}$ and $c_{\nu}$. Having values of $\gamma_{\nu}$ and $c_{\nu}$ for all $\nu \in \mathcal P(\mathcal{A}_{\neg\varphi})$, one can simply utilize results of Theorem~\ref{upper_bound} and Corollary~\ref{lower_bound} to compute a lower bound on the probability of satisfying the given specification.}
	\begin{remark}
		Under the assumption that sets $X,X_0$, and $X_1$ in Theorem \ref{barrier1} are compact and input set $U$ is finite, one can utilize counterexample guided inductive synthesis (CEGIS) approach to search for barrier control functions for more general nonlinear functions $f,h,g_1,g_2,r_1$, and $r_2$ in \eqref{eq:S}. For more detailed discussion on CEGIS approach, we kindly refer interested readers to the algorithm in \cite[Section 5.3.2.]{jagtap2019formal}.
	\end{remark}
	\textbf{Computational Complexity:} The number of triplets and hence the number of control barrier functions needed to be computed are bounded by $|Q|^3$, where $|Q|$ is the number of states in DFA $\mathcal{A}$. However, this is the worst-case bound and in practice, the number of control barrier functions is much smaller. In the case of sum-of-squares optimization approach, the computational complexity of finding polynomial control barrier functions depends on both the degree of polynomials and the number of state variables. One can easily see that for fixed polynomial degrees, the required computations grow polynomially with respect to the dimension of the augmented system. For the CEGIS approach, due to its iterative nature and lack of guarantee on termination, it is difficult to provide any analysis on the computational complexity.
	\section{Case Study}
	We consider a nonlinear Moore-Greitzer jet engine model in no-stall mode
	\cite{krstic1995lean} as a partially observed jump-diffusion systems by adding noise and jump terms which is given by: 
	\begin{equation*}\begin{aligned}
		\diff  \xi_1&=( -\xi_2-\frac{3}{2}\xi_1^2-\frac{1}{2}\xi_1^3)\diff t+0.2\diff W_{11t}+ 0.9\diff P_t,\\ 
		\diff \xi_2&=( \xi_1-\upsilon)\diff t +0.06\diff W_{12t},\\
		\diff y&=\xi_2\diff t +0.06\diff W_{2t},
    \end{aligned}	\end{equation*}
	where $\xi=[\xi_1,\xi_2]^T$, $\xi_1=\Phi-1$, $\xi_2=\Psi-\psi-2$, $\Phi$ is the mass flow, $\Psi$ is the pressure rise, and $\psi$ is a constant. Terms $W_{11t}, W_{12t}$, and $W_{2t}$ denote the standard Brownian motions and $ P_t$ denotes the Poisson process with rate $\lambda=5$. We consider a compact state set $X=[-1, 3]\times[-4,4]$ and regions of interest $X_0=[0,1]\times[-1,1]$, $X_1=[-1,-0.2]\times[-4,-2.5]$, $X_2=[1, 3]\times[2,4]$, and $X_3=X\setminus(X_0\cup X_1\cup X_2)$. The set of atomic propositions is given by $\Pi=\{p_0,p_1,p_2,p_3\}$ with labeling function $L(x_j)=p_j$ for all $x_j\in X_j$, $j\in\{0,1,2,3\}$. The objective here is to compute a control policy that provides a lower bound on the probability that the trajectories of the system satisfy the specification given by the accepting language of the DFA $\mathcal{A}$ given in Figure \ref{fig:DFA} over finite time-horizon $[0,T=10)$. Language of $\mathcal{A}$ entails that if we start in $X_0$ then the system will always stay away from $X_1$ or $X_2$. The corresponding DFA $\mathcal{A}^c$ accepting complement of $\mathcal{L}(\mathcal{A})$ is shown in Figure \ref{fig:DFA}. Following Subsection \ref{subsection:runs}, we only need to compute a control barrier function corresponding to triplet $(q_0,q_1,q_2)$.
	
	Now with an estimator gain in \eqref{eq:estimator1} as $K=[6.1394, 7.8927]^T$,
	we use SOSTOOLS and SeDuMi to compute a sum-of-squares polynomial control barrier function $B(x,\hat{x})$ of order $4$, sum-of-square polynomials $\psi_0(x,\hat{x})$, $\psi_1(x,\hat{x})$, $\psi(x,\hat{x})$ of order $4$, with total $1125$ coefficients resulting in a computation time of about $15$ minutes. The corresponding controller of order 2 is obtained as follows:
	\begin{equation}\label{eq:u_example}
	\mathsf{u}(\hat{x})=0.7321\hat{x}_1-1.8612\hat{x}_1\hat{x}_2-1.4356\hat{x}_2. 
	\end{equation}
	The values of $\gamma=0.099$ and $c=1\times 10^{-5}$ are obtained using bisection method resulting in $\mathbb{P}\{\sigma(\xi_{a\upsilon}) \models \mathcal{A}\}\geq 0.89$ for all $x_0 \in L^{-1}(p_0)$, as discussed in Subsection \ref{section:Computation}. One can see that only one controller is enough for enforcing the specification, thus we do not need any switching mechanism. Figure \ref{fig:realizations} shows a few trajectories starting from different initial conditions under the control policy \eqref{eq:u_example}.
	\begin{figure}
		\centering
		\includegraphics[scale=0.12]{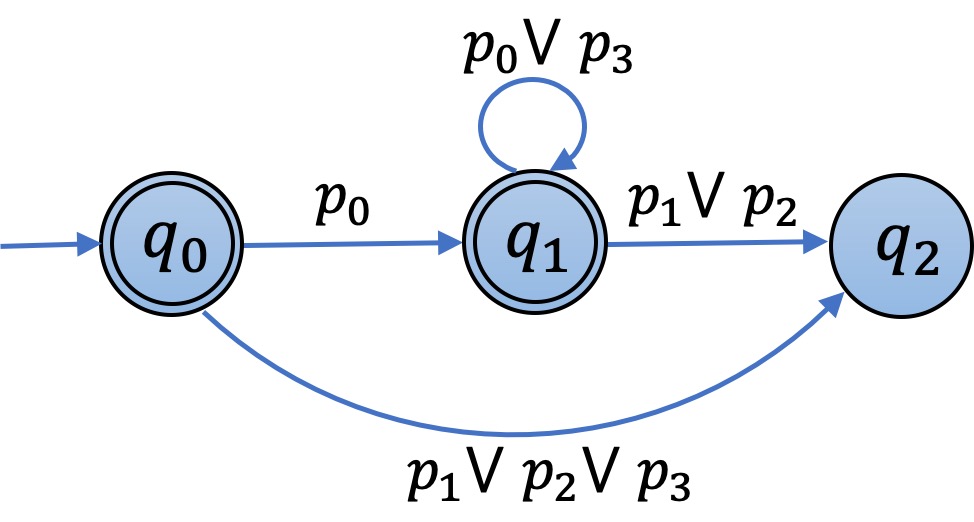}\hspace{0.2em} \includegraphics[scale=0.12]{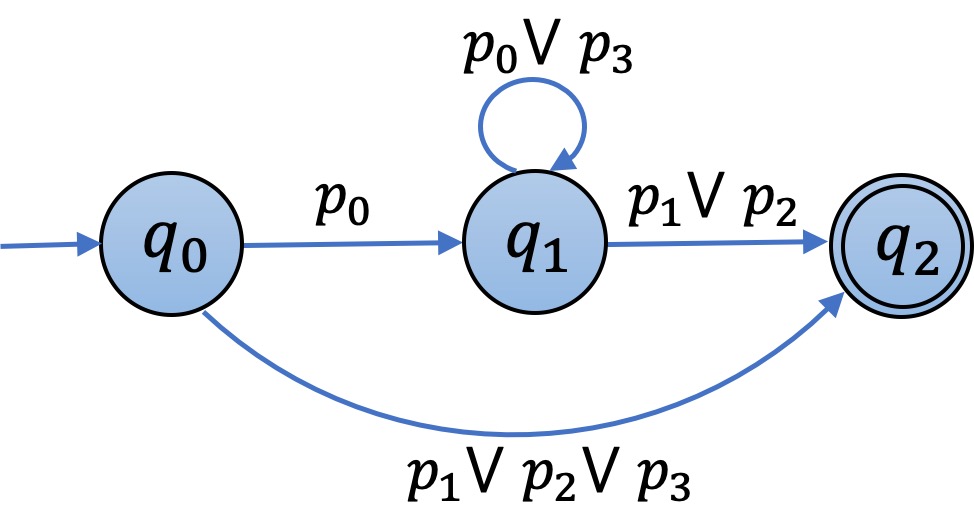}
		\caption{The DFA $\mathcal{A}$ representing specification (left) and the DFA $\mathcal{A}^c$ representing complement of $\mathcal{A}$ (right).}
		\label{fig:DFA}
	\end{figure}
	\begin{figure}
		\centering
		\includegraphics[scale=0.33]{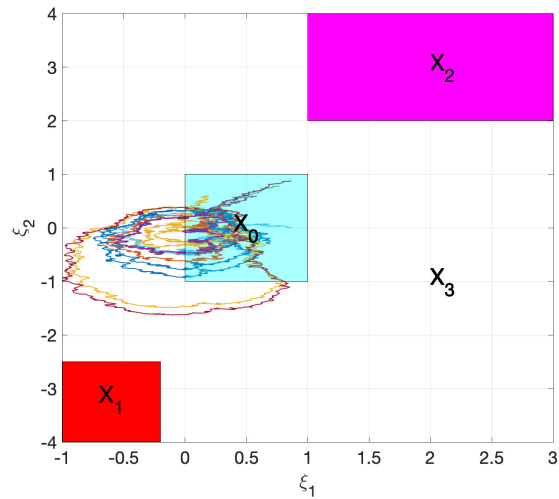} 
		\caption{A few closed loop trajectories starting from different initial conditions in $X_0$ under controller \eqref{eq:u_example}.}
		\label{fig:realizations}
		\vspace{-0.5cm}
	\end{figure}


	\section{Conclusions}
	In this paper, we proposed a discretization-free approach for the formal controller synthesis of partially observed jump-diffusion systems. The proposed method computes a hybrid control policy together with a lower bound on the probability of satisfying complex temporal logic specifications given 
		by the accepting language of DFA $\mathcal{A}$ over a finite-time horizon. This is achieved by constructing control barrier functions over an augmented system consisting of both the system and the estimator. As a result, the probability bound is computed without requiring any prior information of estimation accuracy. 

	\bibliographystyle{IEEEtran}
	
	\bibliography{bibliography.bib}

\end{document}